\documentclass[11pt]{article}
\usepackage{color,euscript,amsfonts,amssymb,epsfig}
\usepackage{graphics}

\newcommand{\eps}         {\varepsilon}
\newcommand{\bareps}	  {\bar{\varepsilon}}

\newcommand{\pow}   	  {\mathrm{pow}}

\newcommand{\norm}[1]     {\| #1 \|}

\newcommand{\nor}[1]	  {\mathrm{{\bf n}}_{#1}}

\newcommand{\cancel}[1]   {}

\renewcommand{\leq}       {\leqslant}

\renewcommand{\geq}       {\geqslant}

\newtheorem{definition}{Definition}
\newtheorem{theorem}{Theorem}[section]
\newtheorem{lemma}{Lemma}[section]
\newtheorem{corollary}{Corollary}[section]

\newtheorem{observation}{Observation}[section]

\newbox\ProofSym
\setbox\ProofSym=\hbox{%
\unitlength=0.18ex%
\begin{picture}(10,10)
\put(0,0){\framebox(9,9){}}
\put(0,3){\framebox(6,6){}}
\end{picture}}

\begin{document}

\begin{titlepage}

\title{Delaunay Edge Flips in Dense Surface Triangulations
\thanks{Research supported by NSF grants CCF-0430735 and CCF-0635008 and RGC, Hong Kong, China
(HKUST 6181/04E).}} 

\author{Siu-Wing Cheng\thanks{Department of Computer Science and
Engineering, HKUST, Clear Water Bay, Hong Kong.  Email: {\tt scheng@cse.ust.hk}} 
\quad\quad Tamal K. Dey\thanks{
Department of Computer Science and Engineering,
The Ohio State University, Columbus, OH 43210, USA.
Email: {\tt tamaldey@cse.ohio-state.edu}}
}

\maketitle

\thispagestyle{empty}

\begin{abstract}
Delaunay flip is an elegant, simple tool to convert
a triangulation of a point set to its Delaunay triangulation.
The technique has been researched extensively for full dimensional 
triangulations of point sets. However, an important case of triangulations
which are not full dimensional is surface triangulations in three
dimensions. In this paper we address the question of converting
a surface triangulation
to a subcomplex of the Delaunay triangulation with edge flips.
We show that the surface triangulations which closely 
approximate a smooth surface with uniform density 
can be transformed to a Delaunay
triangulation with a simple edge flip algorithm.
The condition on uniformity becomes less
stringent with increasing density of the triangulation. If the condition
is dropped completely, the flip algorithm still terminates
although the output surface triangulation becomes ``almost Delaunay"
instead of exactly Delaunay. 
\end{abstract}

\end{titlepage}

\section{Introduction}
\label{sec:introduction}
The importance of computing Delaunay triangulations of point sets
in applications of science and engineering cannot be overemphasized.
A number of different algorithms have been proposed for their
computations~\cite{BKOS00,BY01,Edel}. Among them flip based algorithms
are most popular and perhaps the most dominant approach in practice.
The sheer elegance and simplicity of this approach make
it attractive to implement. 

Since the introduction of Delaunay flips by
Lawson~\cite{Law77} for point sets in two dimensions, several important
extensions have been made including higher dimensional point sets.
Given any triangulation of the convex hull of a point set in
two dimensions, it is known that Delaunay edge flips can convert
the triangulation to the Delaunay triangulation. The rule for 
Delaunay edge flips is simple. First, check
locally if the circumscribing ball of a triangle $t$ contains a vertex
of another triangle $t'$ sharing an edge $e$ with it. If so,  
replace $e$ with the other diagonal edge contained in the
union of $t$ and $t'$. An elegant result is that this process
terminates with the output as the Delaunay triangulation~\cite{BKOS00,BY01,Edel}.
In higher dimensions, the edge flips can be naturally
extended to bi-stellar flips~\cite{Edel}.
However, the approach extended by bi-stellar
flips does not work in higher dimensions. Already in three dimensions
there are examples where the flipping process can get stuck~\cite{Edel}.
Notwithstanding this negative result, Joe~\cite{Joe} showed 
how to construct Delaunay triangulations by inserting points
incrementally and applying bi-stellar flips after each point insertion.
Edelsbrunner and
Shah~\cite{ES96} extended this result to higher dimensional point sets
and to weighted Delaunay triangulations. Recently, Shewchuk~\cite{Shew05}
showed that a combination of flips and some other local operation
called star splay can convert an ``almost Delaunay" triangulation
to the Delaunay one quite efficiently. 

All the aforementioned results deal
with full dimensional triangulations of a point set. 
An important case of triangulations which are not
full dimensional is surface triangulations
in three dimensions. 
Given the increasing demand of computing 
surface triangulations that are sub-complexes of Delaunay
triangulations~\cite{AB98,BO03,Dey06}, 
it is natural to ask if a surface triangulation
can be converted to a Delaunay one by edge flips 
and, if so, under what conditions.
Such a flip algorithm will be useful in many applications. 
For example,
in geometric modeling, shapes are often represented with
subdivision surfaces~\cite{ZS99} or with isosurfaces~\cite{LC87}.
These surfaces are not necessarily Delaunay. If
one can convert these surfaces to a Delaunay one, a number
of tools that exploit Delaunay properties can be used
for further processing. 

In this work we address the question of Delaunay flips
in surface triangulations. Notice that our goal is to
convert a surface triangulation embedded in $\mathbb{R}^3$
to another surface triangulation which is a sub-complex
of the three dimensional Delaunay triangulation of the
vertex set. This is different from the framework considered
elsewhere~\cite{BS05,DZM07} where
a triangulated surface endowed with a flat metric is converted into
an intrinsic Delaunay triangulation comprised of
simplices (not necessarily planar)
embedded in the surface.
In this case the embedding of the surface in $\mathbb{R}^3$
does not play any role whereas in our case the positions
of the vertices in $\mathbb{R}^3$ determine the Delaunay
flips. 

It turns out that dense surface
triangulations are amenable to a simple edge flip algorithm.
A triangulation is dense if it approximates a smooth surface
with sufficient resolution.
We show that a dense triangulation can be flipped to a Delaunay
triangulation if the density is uniform in some sense.
The condition on uniformity depends on the density. The higher
the density, the less stringent is the condition.
The practical implication of this result is that reasonably dense
triangulations can be converted to Delaunay triangulations
with a simple edge flip algorithm. Such dense 
triangulations are numerous in practice. Subdivision and isosurface
meshes are two such examples. Triangulations of moving
vertices offer another such example~\cite{CD08,Shew05}.
In fact, the results in this paper have been used for
a recent algorithm on maintaining deforming meshes with
provable guarantees~\cite{CD08}. What happens if we do not have
the uniformity condition? We show that the flip algorithm still terminates
but the output surface may not be Delaunay. Nonetheless,
this surface is ``almost Delaunay" in the sense that the diametric 
ball of each triangle shrunk by a small amount remains empty.
Bandyopadhyay and Snoeyink~\cite{BS04} showed the usefulness of
such approximate Delaunay triangulations in molecular modeling.
Because of the approximate emptiness properties of the
circumscribing balls of the triangles, these approximate
Delaunay triangulations may find other applications where
exact Delaunay triangulations are not required. 

\section{Preliminaries}

\subsection{Definitions and results}

We need some definitions and results from $\varepsilon$-sampling theory~\cite{AB98}.  
Let $\Sigma\subset\mathbb{R}^3$ be a smooth compact surface without boundary. The
\emph{medial axis} is the set of centers of all maximally empty balls. The {\em
reach} $\gamma(\Sigma)$ of $\Sigma$ is the infimum over Euclidean distances of all
points in $\Sigma$ to its medial axis. This is also the infimum of the
local feature size of $\Sigma$ as defined by Amenta and Bern~\cite{AB98}. 

A surface triangulation $T$ is a finite simplicial $2$-manifold 
embedded in $\mathbb{R}^3$. We say $T$ is a triangulation of 
a surface $\Sigma$ if vertices of $T$
lie in $\Sigma$ and
its underlying space $|T|$ is homeomorphic to $\Sigma$. 
The triangulation $T$ 
has a {\em consistent orientation}
with $\Sigma$ if the oriented normal of each
triangle makes at most $\frac{\pi}{2}$ angle with the oriented
normals of $\Sigma$ at the vertices. 
For a triangle
$t\in T$, let $\rho(t)$ denote the circumradius of $t$.

\subsection{Uniform dense triangulations}
\begin{definition}
A triangulation $T$ of a surface $\Sigma$ is $\eps$-dense if 
each triangle
$t\in T$ has $\rho(t)\leq \eps\gamma(\Sigma)$ and $T$ has
a consistent orientation with $\Sigma$. 
Furthermore, for $\delta<1$, if 
any two vertices in $T$ has distance more than $\delta\eps\gamma(\Sigma)$,
$T$ is called $(\eps,\delta)$-dense or $\delta$-uniform and $\eps$-dense.
\end{definition}

We use notation $B(c,r)$ to denote a ball with center $c$ and radius $r$.
A {\em circumscribing} ball of a triangle $t\subset\mathbb{R}^3$
is any ball that has the vertices of $t$ on its boundary.
The diametric ball $D_t$ is the smallest such ball;
$D_t=B(c,\rho(t))$ where $c$ is the circumcenter
of $t$. We say a vertex
$v\in T$ {\em stabs} a ball $B$ if $v$ lies inside $B$. 
If $t$ shares
an edge, say $pq$, with a triangle $t'=pqs$, then $s$ is a {\em neighbor}
vertex of $t$. Clearly, each triangle has three neighbor vertices. 
 
\begin{definition}
A triangle $t\in T$ is {\em stabbed}
if $D_t$ is stabbed by a vertex of $T$. 
We say $t$ is {\em locally} stabbed
if the stabbing vertex is one of the three neighbor
vertices of $t$ (Figure~\ref{stab}).
\end{definition}
Theorem~\ref{thm1} and Theorem~\ref{thm2} are main results for
uniform dense triangulations.
\begin{theorem}
For $\delta=2\sin 24\eps$ and 
$\eps<\frac{\pi}{72}$, 
any $(\eps,\delta)$-dense triangulation has a stabbed
triangle if and only if it has a locally stabbed triangle.
\label{thm1}
\end{theorem}
Notice that the bound on $\delta$ is $O(\eps)$. 
This implies that dense triangulations
require only mild restrictions on its uniformity.

We will flip edges of dense triangulations to make it Delaunay.
Suppose $pq$ is an edge in a surface triangulation $T$. Flipping
$pq$ means replacing two triangles, say $pqr$ and $pqs$,
incident to $pq$ in $T$ by the triangles $prs$ and $qrs$.
If the new triangulation is $T'$ we write $T\stackrel{pq}{\rightarrow}T'$.
A triangulation $T$ is {\em flipped} to a triangulation $T'$
if there is a sequence of edge flips so that
$T=T_0\stackrel{e_0}{\rightarrow}T_1\stackrel{e_1}{\rightarrow}\cdot\cdot
\stackrel{e_{k-1}}{\rightarrow} T_k=T'$.

\begin{definition}
An edge in a surface triangulation is called {\em flippable} if it is
incident to a locally stabbed triangle.
\end{definition}

One can devise an easy algorithm to convert a 
$(\eps,\delta)$-dense triangulation
to a Delaunay triangulation using Theorem~\ref{thm1}. Simply flip
any flippable edge existing in the triangulation. 
If one can prove that this flip algorithm terminates and
$(\eps,\delta)$-density is maintained after each flip, 
we will have an algorithm to flip an $(\eps,\delta)$-dense triangulation
to a triangulation that does not have any stabbed triangle.
This means each triangle in the new triangulation has its 
diametric ball that does not contain any vertices of $T$.
In other words, this new triangulation is Delaunay.
Actually, the ``Delaunayhood" is stronger since not only
does each triangle have an empty circumscribing ball but the
ball can be chosen to be its diametric ball. Extending the
notion of Gabriel graphs of a set of points in two dimensions,
we call such a triangulation Gabriel. 

\begin{theorem}
For $\delta=2\sin 24\eps$ and 
$\eps<\frac{\pi}{72}$, 
any $(\eps,\delta)$-dense triangulation 
can be flipped to a Gabriel triangulation.
\label{thm2}
\end{theorem} 

Again, notice that the condition on uniformity becomes relaxed with
increasing density.

\subsection{Dense triangulations}
It is natural to seek similar results for dense triangulations that
are not necessarily uniform.
It turns out that such
triangulations can be flipped to almost Delaunay triangulations
but not necessarily to Delaunay triangulations. To prove this
result we will need some generalizations of the concept of
stabbing as well as Delaunay triangulations.
We denote a ball $B=B(c,r)$ shrunk by $\alpha$ as $B^{\alpha}$,
that is, $B^{\alpha}=B(c,r-\alpha)$.
With this definition, $D^{\alpha}_t$ denotes the diametric
ball of $t$ shrunk by $\alpha$.

\begin{definition}
A surface triangulation $T$ is $\alpha$-Gabriel if for each triangle
$t\in T$, the shrunk diametric ball $D_t^{\alpha}$ contains
no vertex of $T$ inside.
\end{definition}

\begin{figure}[ht]
\begin{center}
\epsfig{file=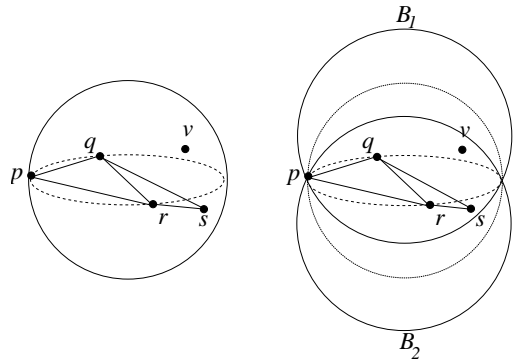,scale=0.3}
\end{center}
\caption{(left) : $pqr$ is stabbed by $v$ and is locally stabbed by $s$;
$pq$ is a flippable edge.
(right) : $B_1$ and $B_2$ are $(\beta)$- and $(-\beta)$-balls of $pqr$ which
is $\beta$-stabbed by $v$ and is locally $\beta$-stabbed by $s$; $pq$ is
a $\beta$-flippable edge.}
\label{stab}
\end{figure}

Let $\nor{t}$ denote the outward normal
of a triangle $t\in T$.
For a triangle $t\in T$ and $\beta\in \mathbb{R}$,
a $\beta$-ball of $t$ is a circumscribing ball of $t$ 
whose center is at $c+\beta \nor{t}$ where $c$ is the
circumcenter of $t$.
Observe that $0$-ball of $t$ is its diametric ball $D_t$. 
For any $\beta\not=0$, there are two balls of 
radius $\sqrt{\rho(t)^2+\beta^2}$,
one is $\beta$-ball and another is $(-\beta)$-ball of $t$. 
See Figure~\ref{stab}.

\begin{definition}
A triangle $t\in T$ is $\beta$-{\em stabbed}
if a vertex of $T$ stabs both $\beta$- and $(-\beta)$-balls
of $t$.
We say $t$ is {\em locally} $\beta$-stabbed
if the stabbing vertex is one of the three neighbor
vertices of $t$.
\end{definition}

Observe that if a triangle $t$ is not $\beta$-stabbed, the intersection
of its $\beta$- and $(-\beta)$-balls cannot contain any vertex of $T$.
This intersection contains the ball $D_t^{\alpha}$, 
$\alpha= \rho(t)+\beta-\sqrt{\rho(t)^2+\beta^2}$, which also cannot
contain any vertex of $T$. Observe that $\alpha\leq \beta$.

\begin{observation}
A surface triangulation $T$ is $\beta$-Gabriel
if it does not have any $\beta$-stabbed
triangle. 
\label{obs-gab}
\end{observation}
 
We prove the following results. 
\begin{theorem}
For $\eps<0.1$, any $\eps$-dense triangulation of a surface
with reach $\gamma$ contains a $\beta$-stabbed
triangle only if it contains a locally $(\beta-88\eps^2\gamma)$-stabbed
triangle.
\label{thm3}
\end{theorem}

Choosing $\beta=88\eps^2\gamma$ we conclude that there is
a $88\eps^2\gamma$-stabbed triangle only if there is a locally
stabbed triangle. Therefore, if one gets rid of all
locally stabbed triangles, there cannot be any $88\eps^2\gamma$-stabbed
triangles. In other words, the triangulation becomes $88\eps^2\gamma$-Gabriel
by Observation~\ref{obs-gab}.

\begin{theorem}
Any $\eps$-dense triangulation of a surface
with reach $\gamma$ can be flipped to a $88\eps^2\gamma$-Gabriel triangulation
if $\eps<0.1$.
\label{thm4}
\end{theorem} 

\subsection{Background results}
The following well known results on normal approximations will be useful in our
analysis. Starting with work of
Amenta and Bern~\cite{AB98}, several versions of these results have been
proved. We pick appropriate ones for our purpose.
Let $\nor{x}$ 
denote the outward unit normal of $\Sigma$ at a point $x\in \Sigma$.  

\begin{lemma}[\cite{CDES01,AD07}]
For any two points $x$ and $y$ in $\Sigma$ such that 
$\|x-y\|\leq \eps \gamma$ for some $\eps \leq \frac{1}{3}$,
$\angle{\nor{x},\nor{y}} \leq \frac{\eps}{1-\eps}$ and
$\angle{\nor{x},(y-x)} \geq \arccos(\frac{\eps}{2})$.
\label{normlem}
\end{lemma} 

Following lemma is an oriented version of a result in~\cite{Dey06}.
Here we use the fact that $T$ has a consistent orientation with 
$\Sigma$. Interestingly, this property does not automatically
follow from triangles being small and vertex set being dense.
\begin{lemma}[\cite{Dey06}]
Let $pqr$ be a triangle in a $\eps$-dense triangulation of a surface.
Assume that $p$ subtends a maximal angle in $pqr$.  Then,
for $\eps<\frac{1}{\sqrt 2}$,
$\angle{\nor{pqr},\nor{p}} \leq \arcsin \varepsilon + \arcsin (\frac{2}{\sqrt
3} (\sin (2\arcsin \varepsilon))$.
\label{trinorm}
\end{lemma}

Combining Lemmas~\ref{normlem} and \ref{trinorm} one obtains
the following corollary.
 
\begin{corollary}
Let $q$ be any vertex in a triangle $pqr\in T$ where $T$ is a
$\eps$-dense triangulation of a surface.
Then, for $\varepsilon<0.1$, $\angle
{\nor{pqr},\nor{q}} \leq 7\varepsilon$.
\label{trinorm-cor}
\end{corollary}

Define the {\em dihedral angle} between two adjacent triangles $pqr$ and
$qrs$ as the angle between their oriented normals, that is,
$\angle{\nor{pqr},\nor{qrs}}$.  An immediate result from
Corollary~\ref{trinorm-cor} is that the dihedral angle between adjacent
triangles in a dense triangulation is small. 

\begin{corollary}
Let $pqr$ and $qrs$ be two triangles in a $\eps$-dense triangulation
of a surface. 
Then, for $\varepsilon<0.1$, $\angle {\nor{pqr},\nor{qrs}} \leq 14\eps$.
\label{dihedral}
\end{corollary}

\section{Flip algorithm}
Our flipping algorithm is very simple.
Continue flipping as long as there is a flippable edge, that is,
an edge incident to a stabbed triangle.\\ 

\noindent
{\sf MeshFlip(T)}
\begin{enumerate}
\item If there is a flippable edge $e\in T$ then flip $e$
else output $T$;
\item $T:=T'$ where $T\stackrel{e}{\rightarrow}T'$; go to step 1.
\end{enumerate}

The first issue to be settled is the termination of {\sf MeshFlip}.
It turns out that this simple flip algorithm terminates if $T$
is a $\eps$-dense triangulation of a surface for 
$\eps<0.1$. 

For convenience we introduce the notion of {\em bisectors} using
power distance. The power distance $\pow(B,x)$ of a point
$x\in\mathbb{R}^3$ to a ball $B=B(c,r)$ is $\|c-x\|^2-r^2$.
For two balls $B_1$ and $B_2$ in $\mathbb{R}^3$, the bisector
$C(B_1,B_2)$ is the plane containing points with equal weighted
distances to $B_1$ and $B_2$.
If $B_1$ and $B_2$ intersect, the bisector $C(B_1,B_2)$
is the plane containing the circle where the boundaries
of $B_1$ and $B_2$ intersect.
For two triangles $pqr$ and $pqs$
sharing an edge $pq$, we write $C_{pq}=C(D_{pqr},D_{pqs})$.
The following lemma establishes symmetry in
stabbing. 

\begin{lemma}
Let $pqr$ and $pqs$ be two adjacent triangles where $s$ stabs
$pqr$. If $\angle{\nor{pqr},\nor{pqs}} < \frac{\pi}{2}$,
$r$ stabs $pqs$.
\label{stab-both}
\end{lemma}
\begin{proof}
It can be shown that the bisector $C_{pq}$ separates $r$ and $s$
if the planes of $pqr$ and $pqs$ make an angle larger than $\frac{\pi}{2}$
or equivalently $\angle{\nor{pqr},\nor{pqs}}<\frac{\pi}{2}$.
Let $C_{pq}^+$ be the half-space supported by $C_{pq}$ and containing
$s$. Clearly, $D_{pqs}$ lies inside $D_{pqr}$ in $C_{pq}^+$ as $s$
is on the boundary of $D_{pqs}$. On the other half-space supported
by $C_{pq}$ which does not contain $s$, $D_{pqr}$ lies inside $D_{pqs}$.
But this half-space contains $r$ which is on the boundary of
$D_{pqr}$. This means $r$ is inside $D_{pqs}$.
\end{proof}

If an edge incident to a stabbed triangle is 
flipped in a triangulation with dihedral angles less than
$\frac{\pi}{2}$, 
the circumradius of each new triangle becomes 
smaller than the circumradius of one of the two triangles 
destroyed by the flip. Actually, this is the key to prove that
flip sequence to get rid of all flippable edges terminate.

\begin{lemma}
Let $T$ be a surface triangulation with dihedral angles
smaller than $\frac{\pi}{2}$. Let $pq\in T$ be an edge
incident to a locally stabbed triangle $pqr$ and $pqs$ be the
other triangle incident to $pq$. We have 
$\rho(qrs) \leq \max\{\rho(pqr),\rho(pqs)\}$
and $\rho(prs) \leq \max\{\rho(pqr),\rho(pqs)\}$.
\label{single-flip}
\end{lemma}
\begin{proof}
We prove the lemma for $\rho(qrs)$. The case for $\rho(prs)$ can be
proved similarly. 
Consider the bisector $C_{pq}$ of $D_{pqr}$
and $D_{pqs}$, see Figure~\ref{prop1}. 
Let $C_{pq}^+$ be the half-space supported by $C_{pq}$
containing $s$ and $C_{rs}^+$ be the half-space supported 
by $C_{rs}$ containing $p$. 

\begin{figure}
\centerline{\epsfig{file=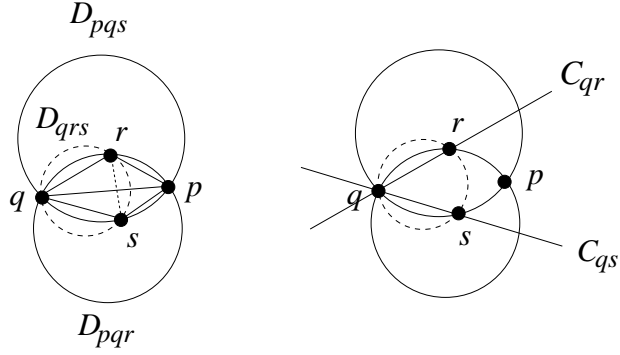,scale=0.55}}
\caption{Lemma~\ref{single-flip}:
triangles after flipping
$pq$ have smaller circumradii.}
\label{prop1}
\end{figure}

By assumption the dihedral angle between $pqr$ and $pqs$ is
at most $\frac{\pi}{2}$. Then, Lemma~\ref{stab-both} applies
to claim that $r$ stabs $pqs$.

Clearly the center of $D_{qrs}$ lies in the union $C_{qr}^+\cup C_{qs}^+$.
First, assume that $C_{qr}^+$ contains the center of
$D_{qrs}$.  Clearly $D_{qrs}\cap C_{qr}^+$ is contained in $D_{pqr}$ as $s$ is
contained in $D_{pqr}$ by the assumption that $s$ stabs $pqr$.  Therefore
$D_{qrs}$ is contained in $D_{pqr}$ in $C_{qr}^+$ which contains the center of
$D_{qrs}$.  This implies that $D_{qrs}$ is smaller than $D_{pqr}$ establishing
the claim. If $C_{qs}^+$ contains the center of $D_{qrs}$ the above
argument can be repeated replacing $D_{pqr}$ with $D_{pqs}$ and $s$ with
$r$.
\end{proof}

Since circumradii of triangles decrease by flipping flippable edges,
triangles still can be oriented consistently with $\Sigma$
and a homeomorphism using closest point map~\cite{ACDL00}
can be established between $\Sigma$ and the new triangulation. 
In sum, the new triangulation satisfies the
conditions for being $\eps$-dense.
\begin{corollary}
If $T\stackrel{e}{\rightarrow} T'$ for a flippable
edge $e$ and $T$ is $\eps$-dense, then
$T'$ is also $\eps$-dense.
\label{density-cor}
\end{corollary}

\begin{lemma}
If $T$ is $\eps$-dense for $\eps<0.1$,
{\sf MeshFlip}$(T)$ terminates.
\label{prop2}
\end{lemma}
\begin{proof}
Let $R_1,R_2,..,R_n$ be the decreasing sequence of the radii of the diametric
balls of the triangles at any instant of the flip process.  First of all, an
edge flip preserves the number of triangles in the triangulation.  An edge
flip may change the entries in this sequence of radii, but not its length.  We
claim that after a flip the new radii sequence $R_1',R_2',...,R_n'$ decreases
lexicographically, that is, there is a $j$ such that $R_i=R_i'$ for all $1\leq
i\leq j$ and $R_{j+1}>R_{j+1}'$.  Let $j+1$ be the first index where
$R_{j+1}\neq R_{j+1}'$. Since each flip
maintains $\eps$-density (Corollary~\ref{density-cor})
the dihedral angles between adjacent triangles remain at most $14\eps$ by
Corollary~\ref{dihedral}. 
This angle is less than $\frac{\pi}{2}$ for $\eps<0.1$.
One can apply Lemma~\ref{single-flip} to each intermediate
triangulation. By this lemma the maximum of the two
radii before a flip decreases after the flip. This means the triangle
corresponding to the radius $R_{j+1}$ has been flipped and its place has been
taken by a triangle whose circumradius is smaller than $R_{j+1}$. So the new
radii sequence is smaller lexicographically.  It follows that the same
triangulation cannot appear twice during the flip sequence.  As there are
finitely many possible triangulations with a fixed number of vertices, the flip
sequence must terminate.
\end{proof}

\section{Uniform dense triangulation}
We prove Theorems~\ref{thm1} and~\ref{thm2} now.
First, we need some 
technical results (Lemmas~\ref{twoside},
and~\ref{stab-free}).
We want to prove that if a vertex stabs the diametric
ball of a triangle, it does not
project orthogonally to a point inside that triangle.
Next lemma is used to prove this fact.

\begin{lemma}
Assume that a vertex $v$ stabs a triangle $pqr$ in a $\eps$-dense 
triangulation of a surface where $\eps<0.1$.
Let $\bar{v}$ be the point in $pqr$ closest to
$v$.  The angle between the segment $v\bar{v}$ and the line of $\nor{pqr}$ is
at least $\frac{\pi}{2} - 26\eps$. 
\label{tech}
\end{lemma}
\begin{proof}
Let $T$ be a $\eps$-dense triangulation of surface $\Sigma$ with reach $\gamma$.
Since $v$ stabs $D_{pqr}$, we have $\norm{p-v} \leq 2\eps \gamma$ which implies
that $\norm{v-\bar{v}} \leq 2\eps\gamma$.  Walk from $v$ towards $\bar{v}$
and let $abc$ be the first triangle in $T$ that we hit.  Let $y$ be the point
in $abc$ that we hit.  (The triangle $abc$ could possibly be $pqr$.) We have
$\norm{v-y} \leq \norm{v-\bar{v}} \leq 2\eps\gamma$.  By $\eps$-density assumption, 
we have $\norm{a-y} \leq 2\eps\gamma$.  It follows that
$\norm{a-v} \leq \norm{a-y} + \norm{v-y} \leq 4\eps\gamma$.  Then,
$\angle{\nor{v},\nor{a}} \leq 8\eps$ by
Lemma~\ref{normlem}, and $\angle{\nor{abc},\nor{a}} \leq 7\eps$ by
Corollary~\ref{trinorm-cor}.  Therefore, $\angle{\nor{v},\nor{abc}} \leq
8\eps + 7\eps\leq 15\eps$.

Let $\ell$ be an oriented line through $v$ and $\bar{v}$ such that $\ell$
enters the polyhedron bounded by $T$ at $y \in abc$ and then exits at $v$.
Assume to the contrary that $\ell$ makes an angle less than $\frac{\pi}{2} -
26\eps$ with $\nor{pqr}$.  Since $\norm{p-v} \leq 2\eps\gamma$,
Lemma~\ref{normlem} and Corollary~\ref{trinorm-cor} imply that $\angle
\nor{v},\nor{pqr} \leq 4\eps + 7\eps\leq 11\eps$.
Thus, $\ell$ makes an angle less than
$\frac{\pi}{2} - 15\eps$ with $\nor{v}$.  Since $\angle \nor{v},\nor{abc} <
15\eps$, $\ell$ must make an angle less than $\frac{\pi}{2}$ with
$\nor{abc}$.  Because $\ell$ enters at $y$ and then exits at $v$, $\angle
\nor{v},\nor{abc}$ is greater than $\pi - (\frac{\pi}{2} - 15\eps) -
\frac{\pi}{2} = 15\eps$, contradicting the previous deduction that $\angle
\nor{v},\nor{abc} < 15\eps$.
\end{proof}
\begin{lemma}

Assume that a vertex $v$ stabs $D_{pqr}$ of a triangle $pqr$ in
a $\eps$-dense triangulation
where $\eps<0.1$.  There exists an edge, say $pq$, such that $r$ and
$v$ are separated by the plane $H_{pq}$ that contains $pq$ and is perpendicular
to $pqr$.
\label{twoside}
\end{lemma}
\begin{proof}
By Lemma~\ref{tech}, $v\bar{v}$ makes a positive angle with the line of
$\nor{pqr}$.  It follows that $v$ does not project orthogonally onto a point
inside $pqr$.  Hence, there exists an edge $pq$ such that $H_{pq}$ separates
$r$ and $v$.
\end{proof}

Next lemma leads to Theorem~\ref{thm1}.
This is where we require bounded aspect ratios of triangles
which ultimately lead to the uniformity condition.
The aspect ratio of a triangle $t$ is the ratio of $\rho(t)$
to its smallest edge length. 

\begin{lemma}
Assume that a vertex $v$ stabs a triangle $pqr$ in a $\eps$-dense
triangulation $T$ where each triangle
has aspect ratio $a<\frac{1}{2\sin 24\eps}$. If $\eps<\frac{\pi}{72}$,
either $pqr$ is locally stabbed or $v$
stabs a triangle $t$ such that $\pow(v,D_{t}) < \pow(v,D_{pqr})$.
\label{stab-free}
\end{lemma}
\begin{proof}
By Lemma~\ref{twoside}, there is a plane $H_{pq}$ through the edge $pq$ and
perpendicular to $pqr$ such that $H_{pq}$ separates $r$ and $v$.  Let $pqs$ be
the other triangle incident to $pq$. If $s$ lies inside $D_{pqr}$, $pqr$ is
locally stabbed and we are done.  So assume that $s$ does not lie inside $D_{pqr}$.
By Corollary~\ref{dihedral} $\angle{\nor{pqr},\nor{pqs}} \leq 14\eps$,
which is less than $\frac{\pi}{2}$ for $\eps<\frac{\pi}{72}$. 
Therefore, $H_{pq}$
separates $r$ and $s$ too.  It means that $v$ and $s$ lie on the same side of
$H_{pq}$;
see Figure~\ref{flipterm}.

Let $C_{pq}$ denote the bisector of $D_{pqr}$ and $D_{pqs}$. Suppose that
$C_{pq}$ contains $v$ and $s$ on the same side.  
It follows that $D_{pqs}$ contains
$D_{pqr}$ inside on this side as $s$ lies outside $D_{pqr}$. Also $v$ lies
inside $D_{pqs}$ since $v$ lies inside $D_{pqr}$. It immediately implies that
$v$ stabs $pqs$ and $\pow(v,D_{pqs})<\pow(v,D_{pqr})$.
Therefore, we can establish the lemma if we prove that $C_{pq}$
contains $v$ and $s$ on the same side. This is exactly where we need
bounded aspect ratios for triangles.

Let $\bar{s}$ and $\bar{v}$ be the orthogonal projections of $s$ and $v$
respectively onto the line of $pq$. 
Consider the following facts.
\begin{enumerate}
\item[(i)] The acute angle between $s\bar{s}$ and $\nor{pqr}$ is 
equal to $\frac{\pi}{2} -
\angle{\nor{pqr},\nor{pqs}}$, which is at least $\frac{\pi}{2} - 14\eps$ by
Corollary~\ref{dihedral}.  
\item[(ii)] The angle between $H_{pq}$ and $C_{pq}$ 
cannot be larger than
$\angle{\nor{pqr},\nor{pqs}}$ which is at most $14\eps$. 
\item[(iii)] We prove that $\angle{\nor{pqr},v\bar{v}}>
\angle{\nor{pqr},\nor{pqs}}=\angle{H_{pq},C_{pq}}$.

\end{enumerate}
The above three facts together imply that $C_{pq}$ contains $v$ and
$s$ on the same side as $H_{pq}$. 
Therefore,
only thing remains to prove is fact (iii).

First, observe that if $\bar{v}$ is the closest point of $v$ in $pq$,
we have by Lemma~\ref{tech}
$$\angle{\nor{pqr},v\bar{v}}\geq \frac{\pi}{2}-26\eps\geq 14\eps\geq 
\angle{\nor{pqr},\nor{pqs}}.$$ 
So, assume the contrary. In that case, the
closest point of $v$ in $pq$ is one of $p$ or $q$. Assume it to be
$p$. Since $\bar{v}$ lies outside
$pq$, the angle $\angle{vpq}$ is obtuse. We claim that
this angle cannot be arbitrarily close to $\pi$. 
In fact, this angle cannot be more than the maximum obtuse angle $pq$ makes
with the tangent plane of $D_{pqr}$ at $p$. Simple calculation 
(Figure~\ref{flipterm}(middle)) shows that
this angle is $\frac{\pi}{2}+\arccos{\|p-q\|/2\rho(pqr)}$ giving
$$
\angle{vpq}\leq \frac{\pi}{2}+\arccos{\frac{1}{2a}}
$$ where $a$ is the aspect ratio of $pqr$. 
Since $D_{pqr}$ contains $v$ inside, $\|v-p\|\leq 2\eps\gamma$.
By Lemma~\ref{normlem} $\angle{\nor{p},vp}\geq\arccos\eps$.
Applying Corollary~\ref{trinorm-cor} we get 
\begin{eqnarray*}
\angle{\nor{pqr},vp}\geq \angle{\nor{p},vp}-\angle{\nor{pqr},\nor{p}}
\geq \arccos{\eps}-7\eps.
\end{eqnarray*}
Let $zp||v\bar{v}$ (Figure~\ref{flipterm}(right)). Then, 
$\angle{v\bar{v},vp}=\angle{vpz}=\angle{vpq}-\frac{\pi}{2}\leq\arccos{\frac{1}{2a}}$.
One has
\begin{eqnarray*}
\angle{\nor{pqr},v\bar{v}}\geq\angle{\nor{pqr},vp}-\angle{v\bar{v},vp}=
 \angle{\nor{pqr},vp}-\angle{vpz}&\geq& 
\arccos{\eps}-7\eps-\arccos{\frac{1}{2a}}\\
&\geq& \frac{\pi}{2}-10\eps-\arccos{\frac{1}{2a}}
\mbox{ for $\eps<\frac{\pi}{72}$}.
\end{eqnarray*}
\begin{eqnarray*}
\mbox{We are now left to show that }
\frac{\pi}{2}-10\eps-\arccos{\frac{1}{2a}}
 & > & \angle{\nor{pqr},\nor{pqs}},\\
\mbox{requiring, }\frac{\pi}{2}-24\eps&>& \arccos{\frac{1}{2a}}\\
or,~a &<& \frac{1}{2\sin 24\eps}.
\end{eqnarray*}
This is precisely the condition required by the lemma
which can be achieved for $\eps<\frac{\pi}{72}$.

\begin{figure}
\centerline{\epsfig{file=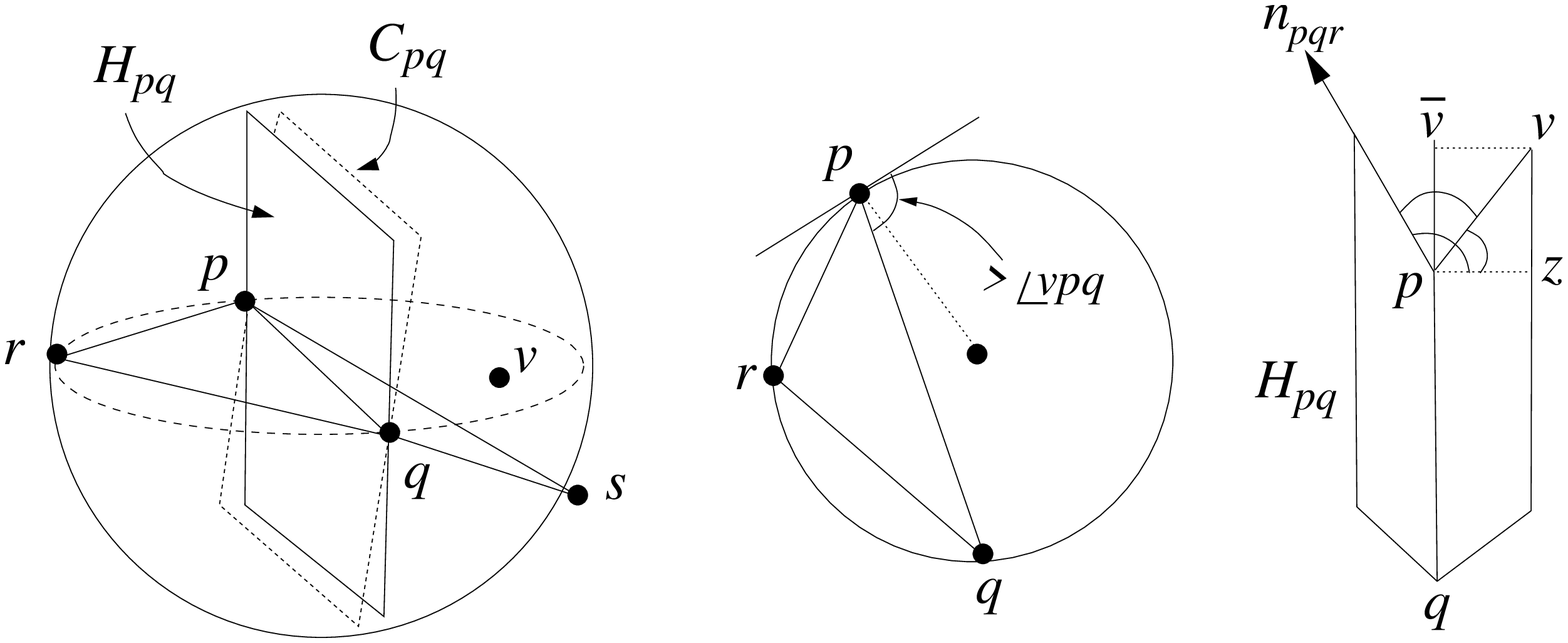,scale=.45}}
\caption{(left) : triangle $pqr$ is stabbed by $v$. Both $v$ and
$s$ lie on the same side of $H_{pq}$ and $C_{pq}$. The case of $v$
being in the thin wedge between $H_{pq}$ and $C_{pq}$ is eliminated
if $pqr$ has bounded aspect ratio. (middle) : the worst case for
angle $\angle{vpq}$. (right): the planes of $H_{pq}$ and $vpq$ make
large angle ensuring $v$ and $s$ are on the same side of $C_{pq}$.}
\label{flipterm}
\end{figure}
\end{proof}

\begin{proof}[Proof of Theorems~\ref{thm1} and ~\ref{thm2}.]
The `if' part of Theorem~\ref{thm1} is obvious. 
To prove the `only if' part,
let $pqr$ be stabbed by $v$. With $\delta=2\sin 24\eps$ aspect ratios
are at most $1/(2\sin 24\eps)$. So, by
Lemma~\ref{stab-free}, either $pqr$ is locally stabbed or $v$ stabs a triangle $t$
where $\pow(v,D_{t})< \pow(v,D_{pqr})$.  In the latter case 
repeat the argument with $t$.  We must reach a locally stabbed triangle
since the power distance of $v$ from the diametric balls 
cannot decrease indefinitely.
For Theorem~\ref{thm2} observe that
maximum circumradius decreases after each flip
and nearest neighbor distance cannot be decreased by flips. So, 
{\sf MeshFlip} maintains $(\eps,\delta)$-density after each flip which is
the only thing remained to be proved. 
\end{proof}

\section{Dense triangulations}
We establish Theorems~\ref{thm3} and~\ref{thm4} in this section. We drop
the uniformity condition, i.e., we assume $T$ is only
$\eps$-dense for some $\eps>0$. We will use the notation
$\pi_t^{\beta}$ to denote the plane parallel to $t$ and
passing through the point $c+\beta\nor{t}$ where $c$ is the
circumcenter of $t$. In other words, $\pi_t^{\beta}$ is the 
diametric plane parallel to $t$ in the $\beta$-ball of $t$.

\begin{lemma}
For $\eps<0.1$, let $T$ be a $\eps$-dense triangulation
of a surface with reach $\gamma$.
A vertex stabs a $\beta$-ball of a triangle $t\in T$ only if
there is a triangle $t'\in T$ with $v$ as a neighbor vertex
and $v$ stabs the circumscribing ball of $t'$
that has center in the plane of $\pi^{\beta}_t$.
\label{alpha-stab}
\end{lemma}

\begin{proof}
Let $B$ be a $\beta$-ball of $t$ stabbed by $w$. Consider the
edges of $T$ lying in $B$ and planes
passing through these edges which are orthogonal to $\pi_t^{\beta}$.
Let $P_e$ denote such a plane passing through the edge $e$.
Let $t=pqr$ and $pq$ be the edge so that $H_{pq}$ separates
$r$ and $w$
according to Lemma~\ref{twoside}. The line segment
$rw$ must cross $P_{pq}=H_{pq}$ and possibly others. Let
$pq=e_1,e_2,...,e_k$ be the sequence of edges so that
$rw$ crosses $P_{e_1},P_{e_2},...,P_{e_k}$ in this order.

Consider two triangles $t_i$ and $t_{i+1}$ incident to any
edge $e_i$ in the sequence $e_1,e_2,...,e_k$.
Let $B_i$ and $B_{i+1}$ be the two balls circumscribing 
$t_i$ and $t_{i+1}$ respectively
and having centers on the plane $\pi_t^{\beta}$. Observe that the
bisector of $B_i$ and $B_{i+1}$ is $P_{e_i}$.
If a vertex of $t_{i+1}$ lies inside $B_i$ we have $t'=t_i$
satisfying the lemma. Otherwise, $B_i$ is contained
in $B_{i+1}$ on the side of $P_{e_i}$ which contains $w$.
So, $\pow(B_i,w)\geq \pow(B_{i+1},w)$.
Since this relation holds for any $i\in [1,k]$, we have
either found the triangle $t'$ satisfying the lemma or
$0>\pow(B,w)=\pow(B_1,w)\geq\pow(B_k,w).$
In the latter case $\pow(B_k,w)$ is negative and hence $B_k$ contains
$w$ inside. The ball $B_k$ circumscribes $t_k$ and has center
in $\pi_t^{\beta}$. It is stabbed by $w$ where
$w$ is a neighbor vertex of $t_k$ satisfying properties of
$t'$ required by lemma.
\end{proof}

\begin{proof}[Proof of Theorem~\ref{thm3}]
Let $t\in T$ be $\beta$-stabbed.
By definition, the $\beta$-ball and $(-\beta)$-ball of $t$
are stabbed by a vertex $w$. Apply Lemma~\ref{alpha-stab} 
to both of these balls. Observe that the planes $(P_e)$ that
we construct in the lemma remain same for both of these
balls. It means that the segment $rw$ in the proof 
crosses same set of planes. In other words, the triangle
$t'$ guaranteed by Lemma~\ref{alpha-stab} remains same. 
Let $B$ and $B'$ be the
two circumscribing balls of $t'$ which have their centers
in $\pi_t^{\beta}$ and $\pi_t^{-\beta}$ respectively.
If we prove that $B$ and $B'$ are larger than
$(\beta-88\eps^2\gamma)$-ball of $t'$, we will be
done since then $t'$ will be locally $(\beta-88\eps^2\gamma)$-stabbed.

Let $c$ and $c'$ be the circumcenters of $t$ and $t'$
respectively. Since $t'$ has an edge in the diametric
ball $D_t$ and all triangles have circumradius less than
$\eps\gamma$, the distance $\|c-c'\|$ is at most $2\eps\gamma$.
We have $\angle{\nor{t},\nor{p}}\leq 7\eps$ and
$\angle{\nor{t'},\nor{s}}\leq 7\eps$ where $p$ and $s$
are vertices of $t$ and $t'$ respectively. Also
the distance between $p$ and $s$ cannot be more than $4\eps\gamma$
which gives $\angle{\nor{p},\nor{s}}\leq 8\eps$.
In all, $\angle{\nor{t},\nor{t'}}\leq 22\eps$ when
$\eps<0.1$. We want to estimate the
distance of $c'$ from the plane of $t$. In the worst
case this distance is
$$\|c-c'\|\sin\angle{\nor{t},\nor{t'}}\leq 4\eps\gamma\sin 22\eps\leq 
88\eps^2\gamma.$$
It means if we choose $\beta>88\eps^2\gamma$, the center
$c'$ lies inside the slab made by offsetting
$\pi_t$ by 
$\beta$ on both sides.
The distance of $c'$ from these
planes is at least $(\beta-88\eps^2\gamma)$. Therefore, with
$\beta>88\eps^2\gamma$ we have $B$ and $B'$ larger than
$(\beta-88\eps^2\gamma)$-ball of $t'$ proving the claim. 
\end{proof}

\begin{proof}[Proof of Theorem~\ref{thm4}]
We apply algorithm {\sf MeshFlip} on the $\eps$-dense
triangulation $T$ of a surface 
whose reach is $\gamma$. 
According to Theorem~\ref{thm3} output
triangulation cannot have any $88\eps^2\gamma$-stabbed triangle.
By Observation~\ref{obs-gab}, the output is $88\eps^2\gamma$-Gabriel.
\end{proof}

Instead of flipping all locally stabbed triangles, one may
flip more conservatively. 
If we go on flipping edges that are incident to $\beta$-stabbed triangles, 
we get a triangulation which is $(\beta+88\eps^2)$-Gabriel.
We flip less edges than 
{\sf MeshFlip} does and hence obtain a worse triangulation
in terms of approximation to Gabriel triangulation.\\

\section{Conclusions}
In this work we showed that a uniform dense surface
triangulation can be flipped to a Delaunay one using simple
Delaunay-like flips. If uniformity condition is dropped, 
we get almost Delaunay surface triangulation. 

This research ensues some open questions.
Can the dense triangulations be flipped to exact Delaunay
triangulation? It is unlikely that such triangulations can be
flipped to exact Gabriel triangulation. It might very well be that 
they cannot be flipped to exact Delaunay triangulations.
Our flip algorithm
converts dense triangulations to almost Gabriel triangulations.
Is it true that such triangulations are actually a weighted
Delaunay triangulation of its vertex set weighted appropriately?
Or, is it possible  to assign weights to the vertices and
carry out edge flips to convert a dense surface triangulation
to a weighted Delaunay one? We plan to address these questions
in future work. 

\end{document}